\documentclass[aps,twocolumn,superscriptaddress,dvipsnames,nofootinbib,amssymb,floatfix,10pt,prx]{revtex4-2}
\usepackage{graphicx}

\usepackage[utf8]{inputenc}
\usepackage{hyperref} %

\usepackage{amsmath}
\usepackage{amsthm}
\usepackage{amsfonts}
\usepackage{amssymb}
\usepackage{amsbsy}
\usepackage{bbm}
\usepackage{bm}
\usepackage{mathtools}

\usepackage{floatrow}

\usepackage{interval}
\intervalconfig{soft open fences}

\usepackage{thmtools}
\usepackage[capitalize]{cleveref}

\usepackage{algorithm}
\usepackage{algpseudocode}
\usepackage{booktabs}
\usepackage{tikz}
\usetikzlibrary{quantikz2}

\declaretheorem[numberwithin=section,name=Theorem]{theorem}
\declaretheorem[name=Proposition, sibling=theorem]{proposition}
\declaretheorem[name=Corollary, sibling=theorem]{corollary}
\declaretheorem[name=Lemma, sibling=theorem]{lemma}
\declaretheorem[name=Definition, sibling=theorem, style=definition]{definition}
\declaretheorem[name=Remark]{remark}

\DeclarePairedDelimiter\rbra{\lparen}{\rparen}
\DeclarePairedDelimiter\sbra{\lbrack}{\rbrack}
\DeclarePairedDelimiter\cbra{\{}{\}}
\DeclarePairedDelimiter\abs{\lvert}{\rvert}
\DeclarePairedDelimiter\norm{\lVert}{\rVert}

\DeclarePairedDelimiter\ave{\langle}{\rangle}

\DeclareMathOperator*{\E}{\mathbb{E}}

\newcommand{\dd}[1]{\mathop{}\!\mathrm{d}#1}

\newcommand{\Poisson}{\mathrm{Poisson}}

\newcommand{\e}{\mathrm{e}}
\newcommand{\ii}{\mathrm{i}}

\newcommand{\vecop}{\operatorname{vec}}

\newcommand{\tr}{\operatorname{tr}}

\newcommand{\TV}{\operatorname{TV}}

\newcommand{\SampCom}{\mathsf{S}}
\newcommand{\Dis}{\mathsf{DIS}}
\newcommand{\Fid}{\operatorname{F}}
\newcommand{\QueryD}{\mathsf{Q}_{\diamond}}
\newcommand{\StableS}[2]{\mathsf{Stable}_{#1}\rbra{#2}}

\begin{document}

\preprint{APS/123-QED}

\title{L\'{e}vy-Khintchine Structure Enables Fast-Forwardable Lindbladian Simulation}

\author{Minbo Gao}%
\email{gaomb@ios.ac.cn}%
\affiliation{Key Laboratory of System Software (Chinese Academy of Sciences),
  Institute of Software, Chinese Academy of Sciences, China}%
\affiliation{University of Chinese Academy of Sciences}

\author{Zhengfeng Ji}%
\email{jizhengfeng@tsinghua.edu.cn}%
\affiliation{ Department of Computer Science and Technology, Tsinghua
  University, Beijing, China }%

\author{Chenghua Liu}%
\email{liuch@ios.ac.cn}%
\affiliation{Key Laboratory of System Software (Chinese Academy of Sciences),
  Institute of Software, Chinese Academy of Sciences, China}%
\affiliation{University of Chinese Academy of Sciences}

\begin{abstract}
  Simulation of open quantum systems is an area of active research in quantum
  algorithms.
  In this work, we revisit the connection between Markovian open-system dynamics
  and averages of Hamiltonian real-time evolutions, which we refer to as
  Hamiltonian twirling channels.
  By applying the L\'evy-Khintchine representation theorem, we
  clarify when and how a dissipative dynamics can be realized using Hamiltonian
  twirling channels.
  Guided by the general theory, we explore Hamiltonian twirling with
  Gaussian, compound Poisson and symmetric stable distributions and their algorithmic implications.
  These give wide classes of Lindbladians that
  can be simulated in $\Theta(t^{1/\alpha})$ Hamiltonian simulation time \emph{without any extra ancilla or other quantum gates} for $1\le \alpha \le 2$. Moreover, 
  we prove that these time complexities are asymptotically \emph{optimal} using an information theoretic approach, which, to the best of our knowledge, is the \emph{first} result of lower bounds on fast-forwarding simulation algorithms.
\end{abstract}

\maketitle

\section{Introduction}%
\label{sec:introduction}

The dynamics of Markovian open quantum systems are most generally described by
Lindblad master equations~\cite{GKS76,Lin76}, which provide a unified framework
for dissipation, decoherence, and measurement-induced effects arising from
system-environment coupling.
An important family of such open-system dynamics is characterized by L\'{e}vy
processes, which generalize Brownian motions and arise naturally in a variety of
physically relevant settings.
Examples include irreversible quantum evolutions due to interactions with
classical Gaussian and Poisson noise~\cite{Hol01}, decoherence governed by
L\'{e}vy stable laws arising from coupling to chaotic subsystems~\cite{KBD99},
and scattering-induced dynamics of center-of-mass degrees of
freedom~\cite{Vac05}, where the interplay between L\'{e}vy processes and
decoherence has been extensively studied.
As these dynamics arise from convolution semigroups of L\'{e}vy distributions,
we refer to them as \emph{(quantum) L\'{e}vy dynamics}.
Remarkably, the corresponding generators fall within the class of covariant
Lindbladians studied in~\cite{Hol95, Hol95a, Hol96, Hol01}, which admit a
quantum L\'{e}vy--Khintchine representation consisting of Hamiltonians, Gaussian
(diffusive) terms, and Poisson (jump) terms derived via quantum stochastic
calculus~\cite{Kho91}.

Given the physical importance of L\'{e}vy dynamics and the rapid development of
quantum computing, it is natural to ask whether such open-system dynamics can be
efficiently simulated on a quantum computer.
General Lindbladian simulation algorithms have been extensively studied in
recent years.
The first such algorithm~\cite{KBG+11} achieves complexity
$O\rbra{t^2/\varepsilon}$ for evolution time $t$ and error $\varepsilon$, which
was later improved to $O\rbra{t^{1.5}/\sqrt{\varepsilon}}$ in~\cite{CL17}.
More recently, algorithms with complexity
$O\rbra{t\operatorname{polylog}\rbra{t/\varepsilon}}$ were proposed using the
linear-combination-of-unitaries framework~\cite{CW17,LW23}.
Additional approaches to Lindbladian simulation, encompassing different models
and techniques, have been explored in~\cite{SHS+21,KSMM22,
  SHS+22,PW23a,PW23,DCL24,DLL24,DVF+24,GKP+25,BM25,CKBG25,PSZZ25,PSW25}.
These simulation algorithms are already nearly optimal as suggested by a
no-fast-forwarding theorem in~\cite{CL17}.
It nevertheless remains an intriguing question to ask whether the special
structure of L\'{e}vy dynamics enables faster simulation algorithms, and whether
tools from the theory of L\'{e}vy dynamics can yield new insights into
Lindbladian simulation.

In this Letter, we answer these questions in the affirmative by introducing the
L\'{e}vy-Khintchine representation theorem and the technique of Hamiltonian
twirling~\cite{Kho91,Vac05,AKMV10} into the study of quantum algorithms for
open-system simulation.
Using the L\'{e}vy-Khintchine representation, we present the precise conditions
under which a classical random source interacting with a quantum system induces
Lindblad dynamics, namely L\'{e}vy dynamics.
We further prove that L\'{e}vy dynamics can be simulated in a
\emph{time-optimal} and \emph{ancilla-free} way.
Interestingly, this includes the recently proposed fast-forwarding
$\widetilde{O}(\sqrt{t})$ algorithm for Lindbladians with a single Hermitian
jump operator~\cite{SGR+25} as a special case using \emph{Gaussian twirling} of
Hamiltonian evolutions.
Furthermore, by purifying the Gaussian twirling channel, we show that the
Lindbladian evolution acts as a controlled phase shift gate on a
continuous-variable ancillary system, therefore implying a continuous-variable
quantum phase estimation algorithm (CV-QPE).
More generally, we establish the existence of Lindbladians whose simulation have
exactly the complexity $\Theta(t^{1/\alpha})$ for every $\alpha\in\interval{1}{2}$,
thereby substantially expanding the known landscape of fast-forwardable
Lindbladian dynamics.
The optimality of these algorithms is supported by matching lower bounds
obtained via a lifting from information-theoretic lower bounds.
Our main results are summarized more formally below.

\begin{table*}[t]
\caption{Infinitely divisible distributions, Lindbladian forms of their corresponding Hamiltonian twirling channels, and Hamiltonian simulation time complexity. \emph{All these Lindbladians can be simulated without ancilla.}}
\label{tab:Hamiltonian-twirling-collection-results}
\centering
\begin{ruledtabular}
\begin{tabular}{@{}lcc@{}}
Distribution $\mathcal \mu_t$ & Lindbladian $\mathcal L(\rho)$ & Time Complexity \\
\midrule
Dirac delta $\delta_{at}$ &
$-\mathrm{i} a[H,\rho]$ &
$\Theta(t)$ \\
\midrule
Gaussian $\mathcal N(0,\sigma^2 t)$ &
$-\tfrac12\sigma^2[H,[H,\rho]]$ &
$\Theta(\sqrt{t})$ \\
\midrule
\begin{tabular}[c]{@{}l@{}}
Symmetric $\alpha$-stable, \\
$\StableS{\alpha}{ct}$ ($1<\alpha\le 2$) 
\end{tabular} &
$\propto \int_{\abs{s}>0} \dd{s} \frac{1}{\abs{s}^{1+\alpha}} \rbra*{\mathrm{e}^{-\mathrm{i}Hs}\rho \mathrm{e}^{\mathrm{i}Hs} - \rho}$ &
$\Theta(t^{1/\alpha})$ \\
\midrule
\begin{tabular}[c]{@{}l@{}}
Compound Poisson, \\
with rate $t$ and discrete jump law $\mu$
\end{tabular} &
$\E_{s\sim \mu}\!\big[\mathrm{e}^{\mathrm{i}Hs}\rho\mathrm{e}^{-\mathrm{i}Hs}\big]-\rho$ &
$\Theta(t)$ \\
\end{tabular}
\end{ruledtabular}
\end{table*}

\noindent
\textit{Simulation via Hamiltonian Twirling.---\/} For a Hamiltonian $H$ and a
probability distribution $\mu$ on $\mathbb{R}$, we examine the
properties of the channel
$\E_{s\sim \mu} \sbra{\e^{\ii H s} \rho \e^{-\ii H s}}$, which we refer
to as a Hamiltonian twirling channel.
This channel is a classical mixture of Hamiltonian evolutions and has been
extensively studied in quantum stochastic calculus (see, e.g.,~\cite{Kho91}) and
in the context of Brownian motions on a Lie group~\cite{AKMV10}.
When $\mu$ is infinitely divisible\footnote{A probability distribution
  is infinitely divisible, if for any $n$, it can be written as the probability
  distribution of the sum of $n$ copies of independent and identically
  distributed random variables. For example, Gaussian, Cauchy and Poisson are infinitely divisible. See~\cite{Sat13} and~\cite{sm25} I.D. for more details.}, 
  the family
${\{ \Phi_{H,\mu_t} \}}_{t\ge 0}$ forms a CPTP semigroup with a
Lindbladian generator that commutes with~$H$.
Conversely, every random unitary Lindbladian dynamic can be represented in the
Hamiltonian twirling form.
Furthermore, the explicit form of the Lindbladian operator can be derived using
the L\'{e}vy-Khintchine representation theorem for infinitely divisible
distributions.
Although the theory of L\'{e}vy dynamics is well-established in physics and
mathematics, its application to quantum algorithms for Lindbladian simulation
represents a new direction.

\begin{figure}
  \centering
  \includegraphics[width=\textwidth]{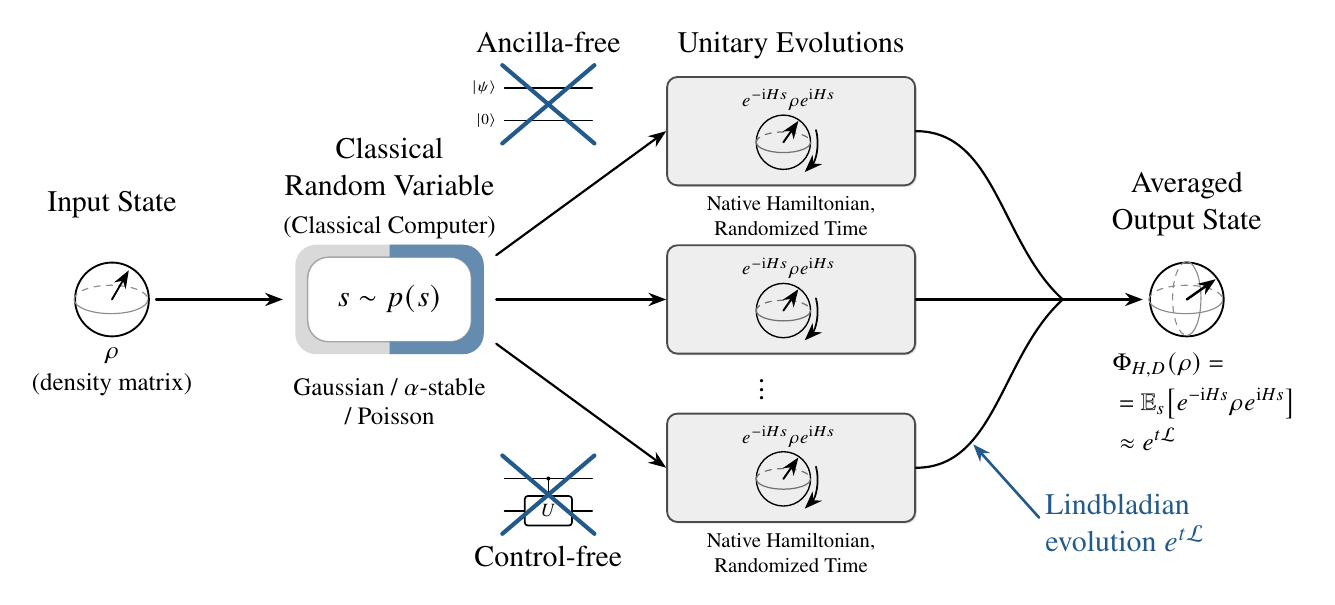}
  \caption{Hamiltonian twirling for Lindbladian simulation.}\label{fig:twirling}
\end{figure}

\noindent
\textit{Ancilla-free Simulation for L\'evy Dynamics.---\/} Based on the exact form
of the Lindbladian corresponding to Hamiltonian twirling channels, we study
several important cases where the underlying distributions being Gaussian,
compound Poisson and $\alpha$-stable, explicitly analyzing the time complexity
of simulating the corresponding dynamics.
For Gaussian and $\alpha$-stable distributions, their induced Hamiltonian
twirling channels result in $O(\sqrt{t})$ and $O(t^{1/\alpha})$ fast-forwarding
Lindbladian simulation algorithms.
Specifically, a Gaussian distribution induces a Lindbladian with a single
Hermitian jump operator $H$, which is also considered in~\cite{SGR+25}.
Achieving similar $O(\sqrt{t})$ simulation time, our Gaussian twirling algorithm
is experimentally much more friendly than the approach in~\cite{SGR+25}, as ours
avoids the use of extra ancillas, Dicke state preparation, and control access to
a dilated Hamiltonian.

\noindent
\textit{Matching Lower Bounds.---\/} To complement the results on the simulation
algorithms obtained via Hamiltonian twirling, we establish matching lower bounds
for the query complexity of the jump operators.
This demonstrates that our algorithms are not only space-optimal, as they are
ancilla-free, but also time-optimal.

Notably, our proof for these lower bounds builds on the sample-to-query lifting
technique~\cite{WZ25,CWZ25,TWZ25} which could help us lift information-theoretic
sample lower bounds to the query complexity of quantum algorithms.
The core intuition is that, through the Schur-channel representation of the
Hamiltonian twirling channel, the task of distinguishing two carefully chosen
diagonal states can be reduced to simulating a Lindbladian whose single
Hermitian jump operator is precisely that state.
This reduction allows us to transfer sample-complexity lower bounds for
state discrimination directly to query lower bounds for Lindbladian simulations.

\section{Revisiting Hamiltonian twirling channels}

Hamiltonian twirling channels arise naturally in several earlier contexts.
In particular, related constructions appear in the theory of quantum stochastic
calculus~\cite{Kho91} and in the study of Brownian motions on Lie
groups~\cite{AKMV10}.
Here we revisit these channels from an algorithmic perspective.
For a Hamiltonian $H$ and a probability distribution $\mu$ on
$\mathbb{R}$, we call the following channel the Hamiltonian twirling channel of
$H$ with respect to $\mu$.
\begin{equation}
  \Phi_{H,\mu}\rbra{\rho}
  = \E_{s\sim \mu} \sbra{\e^{\ii H s} \rho  \e^{-\ii H s}}.
\end{equation}

In this representation, Hamiltonian twirling channels admit efficient
implementation on quantum devices, as they require only Hamiltonian evolutions.
This motivates a closer study of their properties.

Expanding $H$ and $\rho$ in $H$'s eigenbasis as
$H = \sum_j \lambda_j \ket{j}\bra{j}$,
$\rho = \sum_{jk}\rho_{jk}\ket{j}\bra{k}$, we compute the effect of
$\Phi_{H, \mu}$ on $\rho$ and show it acts as a Schur channel of the
form
\begin{equation}\label{eqn:Schur-form-Hamiltonian-twirling}
  \Phi_{H,\mu} \rbra*{\rho} = \sum_{jk}
  \hat{\mu}\rbra*{\lambda_j-\lambda_k} \rho_{jk} \ket{j}\bra{k},
\end{equation}
where
$\hat{\mu}(z) = \E_{x\sim \mu}[\mathrm{e}^{\mathrm{i}zx}]$ is
the characteristic function of the distribution $\mu$.
  
We are particularly interested in a family of these channels
$\{\Phi_{H, \mu_t}\}$ induced by a family of distributions
$\{\mu_t\}$.
A natural question is what conditions imply $\{\Phi_{H, \mu_t}\}$ being
Markovian, i.e., of the form $\mathrm{e}^{t \mathcal{L}}$ of a Lindbladian
$\mathcal{L}$.
\Cref{eqn:Schur-form-Hamiltonian-twirling} implies that the characteristic
functions should be multiplicative, i.e.,
$\hat{\mu}_{s+t} = \hat{\mu}_{s} \hat{\mu}_{t}$, which
we prove is the exact condition.
This means $\mu_t$ is infinitely divisible for every $t$, which admits a
L\'{e}vy-Khintchine representation of its characteristic function as
\begin{equation*}%
  \begin{aligned}
    \hat{\mu}\rbra*{z}
    = \exp \Big[& -\frac{1}{2}\sigma_t^2 z^2 + \ii\gamma_t z \\
                & + \int \rbra*{ \e^{\ii z\cdot x} -1 - \ii
                  z\cdot x 1_{\abs{x}\le 1}\rbra*{x} }\nu_t\rbra*{\dd{x}} \Big],
  \end{aligned}
\end{equation*}
where $\sigma_t\ge 0$, $\nu_t$ is a L\'{e}vy measure, and $\gamma_t\in \mathbb{R}$.
This corresponds to general results in quantum stochastic calculus~\cite{Kho91}
and Brownian motions on Lie groups~\cite{AKMV10}.
  
Moreover, using the L\'{e}vy-Khintchine representation, we explicitly compute
the corresponding Lindblad operator $\mathcal{L}$, which has effective
Hamiltonian term
$H_{\text{eff}} = \rbra*{\int_{-1}^{1} s \; \nu \rbra*{\dd{s}}-\gamma} H$, and
jump operators $L_0 = \sigma H$,
$L_{s} = \e^{\ii H s} \sqrt{\frac{\dd{\nu}}{\dd{t}}\rbra{s}}$ for
$s\in \mathbb{R}$.

\section{Ancilla-free Lindbladian Simulation for L\'evy Dynamics}

\noindent
\textit{Gaussian Twirling.---\/} The L\'{e}vy representation of a Gaussian
distribution $\mathcal{\mu}_{t}=\mathcal{N}(0,t)$ has a particular simple form,
namely $\hat{\mu}_{t}(z) = \exp\rbra*{-\frac{t z^2}{2}}$.
Then, by the above computation, we know the corresponding Lindbladian has only
one single jump operator $\sqrt{t}H$.
To put another way, the Hamiltonian twirling channels induced by (mean-zero)
Gaussian distributions efficiently implements time $t$ evolution of the
Lindbladian $\mathcal{L}(\rho)=H\rho H-\tfrac12\{\rho,H^2\}$.

When the simulation time needs to be upper bounded, one can
use a truncated Gaussian distribution instead. Utilizing a concentration bound
for Gaussian distributions, one can show that at most $\sqrt{2t \log\rbra*{\frac{4}{\varepsilon}}}$ time evolution suffices to implement $\mathrm{e}^{t\mathcal{L}}$ within diamond norm error $\varepsilon$.

\begin{algorithm}[htb]
\caption{Gaussian Distribution Twirling with a Cutoff}%
\label{alg:gaussian-twirl-cutoff}
\begin{algorithmic}[1]
\Require{} Time parameter $t > 0$, accuracy $\varepsilon \in (0,1)$, and Hamiltonian $H$.
\Ensure{} The algorithm implements a channel close to $\e^{t \mathcal{L}}$ to within diamond distance $\varepsilon$.
\State{} Set a cutoff $S \coloneqq \sqrt{2t \log\rbra*{\frac{4}{\varepsilon}}}$.
\State{} Sample $s$ from the truncated normal distribution $q_{t,S}$ with density
\begin{equation*}
q_{t,S}(s) \propto
\begin{cases}
\e^{-s^2/2t}, & \text{if } s \in \interval{-S}{S}, \\
0, & \text{otherwise}.
\end{cases}
\end{equation*}
\State{} Apply the unitary $U_s \coloneqq \e^{-\ii Hs}$ to the system.
\State{} Discard the sampled value $s$.
\end{algorithmic}
\end{algorithm}

Therefore, we arrived at the following conclusion: Simulating
$\mathcal{L} = -\frac{1}{2}\sbra{H, \sbra{H, \cdot}}$ for evolution time $t$ to
within diamond distance error $\varepsilon$ can be done by
Algorithm~\ref{alg:gaussian-twirl-cutoff} using only
$O(\sqrt{t\log (1/\varepsilon)})$ time Hamiltonian simulation of $H$
\emph{without any extra ancilla or other quantum gates}.
This result can be easily generalized to the case where Lindbladians have
commuting Hermitian jump operators by sequential simulation, matching the
results in~\cite{SGR+25} while avoiding using extra ancillas or control.

\noindent
\textit{CV Quantum Phase Estimation.---\/} Following the approach
of~\cite{SGR+25}, our Lindbladian simulation algorithm also gives rise to a
continuous-variable quantum phase estimation.
Consider a Hamiltonian $H$ with its eigendecomposition
$H = \sum_j \lambda_j \ket{\psi_j}\bra{\psi_j}$ and the initial pure state
$\ket{\psi_0}$.
Let $p_t(s)$ denote the probability density function of some (infinitely
divisible) distribution $\mu_t$ with characteristic function $\hat{\mu}(z)$.
Then, the purified version of our Hamiltonian twirling channel works as follows:

First, the algorithm starts with the initial state
$\ket{\phi_0}_{AB} = \int_{\mathbb{R}} \sqrt{p_t(s)}\ket{s}_A\dd{s} \otimes \ket{\psi_0}_B$,
and then applies the controlled Hamiltonian evolution
$U_{H}\coloneqq \int_{\mathbb{R}} \ket{s}_A\bra{s} \otimes \mathrm{e}^{-\ii H s} $
on $\ket{\phi_0}_{AB}$.
One could obtain the state after the corresponding Lindbladian evolution
$\mathrm{e}^{t \mathcal{L}}$ by tracing out the system $A$.

Now, we consider how to extract information about eigenvalues of $H$ from the
purified version $U_H \ket{\phi_0}_{AB}$.
Noting that
\begin{equation*}
  \begin{aligned}
    U_H \ket{\phi_{0}}_{AB}
    & = \int_{\mathbb{R}} \sqrt{p_t(s)} \ket{s}_A \otimes
      \mathrm{e}^{-\ii Hs} \ket{\psi_0} \dd{s} \\
    &= \int_{\mathbb{R}} \sqrt{p_t(s)} \mathrm{e}^{-\ii \lambda_0 s}
      \ket{s}_A \dd{s} \otimes \ket{\psi_0},
  \end{aligned}
\end{equation*}
we know the final state is
$\ket{\widetilde{\phi}_{\lambda_0}} \otimes \ket{\psi_0}$ where
$\ket{\widetilde{\phi}_{\lambda_0}} = \int_{\mathbb{R}} \sqrt{p_t(s)} \mathrm{e}^{-\ii \lambda_0 s} \ket{s}_A \dd{s}$.

To summarize, \emph{the purified Lindbladian dynamics procedure acts as a control phase-shift gate} on the ancilla system with the initial state $\int_{\mathbb{R}} \sqrt{p_t(s)} \ket{s}_A \dd{s}$.

In the case of Gaussian distributions, one can apply the idea of conjugate
observable (see, for instance,~\cite{WPG+12}) as follows.
Let $\hat{s} = \int_{\mathbb{R}} s\ket{s}\bra{s} \dd{s}$, and $\hat{k}$ denote
an observable with $\sbra{\hat{s}, \hat{k}} = \mathrm{i}$, such that the
eigenvectors $\ket{k}$ of $\hat{k}$ satisfies
$\braket{s}{k} = \frac{1}{\sqrt{2\pi}} \mathrm{e}^{\mathrm{i}sk}$.
Expanding $\ket{\widetilde{\phi}_{\lambda_0}}$ in basis $\ket{k}$, we know
$\ket{\widetilde{\phi}_{\lambda_0}} = \int_{\mathbb{R}} \rbra*{\frac{2t}{\pi}}^{1/4}
\exp\rbra*{-t(k+\lambda_{0})^{2}} \ket{k} \dd{k}$.
Therefore, the measure result in basis $\ket{k}$ follows the distribution
$\mathcal{N}(-\lambda_0, \frac{1}{4t})$, and one can estimate the mean value of
this normal distribution to compute $\lambda_0$.
Therefore, by adjusting $t$, we obtain an $O\rbra{\varepsilon^{-1}}$-time CV-QPE
algorithm giving an estimate of error less than $\varepsilon$ with probability at least $2/3$.

Viewed in this way, the connection between quantum phase estimation and~\cite{SGR+25}'s Lindbladian simulation method could be seen as a discrete approximation of the above physical picture.

\noindent
\textit{Stable Distributions Twirling.---\/} Next, we turn our attention to
stable distributions, which are natural generalizations of Gaussian
distribution.
A distribution is stable if a linear combination of two independent random
variables with this distribution follows the same distribution up to location
and scale parameter.
For a symmetric $\alpha$-stable distribution $\StableS{\alpha}{c}$ (where
$1< \alpha \le 2$), its characteristic function can be written as
$\hat{\mu}_c(z) = \exp \rbra{-c\abs{z}^{\alpha}}$.
We refer readers to Supplementary Materials~\cite{sm25} for more details.

By implementing the Hamiltonian twirling channels
$\Phi_{H,\StableS{\alpha}{ct}}$, we can simulate the dynamics
\begin{equation}
  \mathrm{e}^{t\mathcal{L}}[\rho] = \sum_{jk} \exp\rbra{-t
    c \abs{\lambda_j-\lambda_k}^{\alpha} } \rho_{jk} \ket{j}\bra{k}
\end{equation}
with Hamiltonian simulation time $O(t^{1/\alpha})$ (see~\cite{sm25} for more
details).

\noindent
\textit{Poisson and Compound Poisson Twirling.---\/} As our last example, we
consider the case where the distribution is Poisson.
For $\mu_t = \Poisson\rbra{t}$, we know its characteristic function is just
$\hat{\mu}_t(z) = \exp\rbra{t\rbra{\exp\rbra{\ii z}-1}}$.
This yields an ancilla-free simulation algorithm with $O(t)$ cost, 
obtained via our Hamiltonian twirling approach, for the Lindbladian
$
  \mathcal{L} \rbra{\rho} = \e^{\ii H} \rho \e^{-\ii H} - \rho$,
that is, a Lindbladian generated by a single unitary jump operator $\e^{-\ii H}$.

By replacing Poisson distribution with general compound Poisson distributions,
the induced Hamiltonian twirling channel simulates the following Lindbladian
dynamics
$\mathcal{L} \rbra{\rho} = \sum_j p_j \e^{\ii H s_j} \rho \e^{-\ii H s_j} - \rho$,
or, a Lindbladian with mixed unitary jump operators.

\noindent
\textit{Summary.---\/} Collecting the above results
in~\Cref{tab:Hamiltonian-twirling-collection-results}, we have proven the
following theorem.

\begin{theorem}\label{thm:algorithm-Lindbladian-complexity-general} 
  For every $1 < \alpha \le 2$, there exist classes of Lindbladian dynamics that admit \emph{fast-forwarded simulation}, in the sense that evolution for time $t$ can be implemented using only $O(t^{1/\alpha})$ Hamiltonian simulation time, without any ancillary systems or additional quantum gates. 
\end{theorem}

\section{Matching Lower Bounds by Sample to Query Lifting}

The perspective of purified Hamiltonian twirling channels provides useful
intuition for establishing lower bounds for the specific Lindbladians considered
in this work.
Conceptually, the key step is the application of the randomized purification
channel $\Lambda^{(n)}_{\textup{purify}}$ introduced in~\cite{TWZ25,GML25,WW25}.
This channel satisfies, for any density operator $\rho$,
$\Lambda^{(n)}_{\textup{purify}}(\rho^{\otimes n}) = \E_{\ket{\rho}} \ket{\rho}\bra{\rho}^{\otimes n}$,
where the expectation is taken over a uniformly random purification $\ket{\rho}$
of $\rho$.
Informally, this construction allows one to obtain a purification of the density
operator without requiring additional oracle queries.
This observation motivates exploiting the connection between purified
Hamiltonian twirling channels and continuous-variable quantum phase estimation
to derive query lower bounds.

To turn the above observations into rigorous lower bounds, we will use the
technique of sample-to-query lifting~\cite{WZ25,WZ25b,CWZ25,TWZ25,GP22,
  GLM+22,GZ25} that builds on the randomized purification protocol.
This technique enables the translation of sample complexity lower bounds into
corresponding query complexity lower bounds, thereby providing a systematic
framework for our analysis.
To be more exact, let $\sigma_1$ and $\sigma_2$ be two density operators on a
Hilbert space.
Let $\SampCom\rbra{\Dis_{\sigma_1,\sigma_2}}$ denote the sample complexity to
distinguish $\sigma_1$ from $\sigma_2$, and
$\QueryD\rbra{\Dis_{\sigma_1, \sigma_2}}$ is used to denote the query complexity
to distinguish them when given a block-encoding access (see Supplementary Materials~\cite{sm25} for more
details about the input model and the results).
Surprisingly, the central results in~\cite{WZ25,CWZ25,TWZ25} establish that
\begin{equation*}
  \QueryD\rbra{\Dis_{\sigma_1, \sigma_2}} =
  \Omega\rbra*{\sqrt{\SampCom\rbra{\Dis_{\sigma_1,\sigma_2}}}}.
\end{equation*}

By combining the sample-to-query lifting technique with our Schur channel
representation (\Cref{eqn:Schur-form-Hamiltonian-twirling}), we can prove lower
bounds for our Lindbladian simulation algorithms.
We demonstrate this technique for a lower bound on the Gaussian twirling
channels as follows, and leave the details of other cases in~\cite{sm25}.

Let $\sigma_1 = \frac{1}{2}I$ and
$\sigma_2 = \rbra{\frac{1}{2}+\frac{1}{c}} \ket{0}\bra{0} +
\rbra{\frac{1}{2}-\frac{1}{c}}\ket{1}\bra{1}$.
Note that, the infidelity between $\sigma_1$ and $\sigma_2$ can be computed as
$\gamma \coloneqq 1- \Fid\rbra{\sigma_1, \sigma_2} \le \frac{4}{c^2}$.
Thus, the sample complexity of discriminating these states is
$\Omega\rbra{c^2}$.
The sample-to-query lifting therefore gives an $\Omega\rbra{c}$ query lower
bound for this problem.
  
However, by choosing $t = 9c^{2}$, we can show that performing Lindbladian
simulation with the jump operator being $\sigma_1$ or $\sigma_2$ for time $t$
for constant rounds suffices to distinguish the two cases with success
probability at least $2/3$ via the following reduction
(\cref{alg:reduction-lifting-Gaussian}).
Plugging $c = \sqrt{t}/3$, we get an $\Omega\rbra{\sqrt{t}}$ query lower bound.

\begin{algorithm}[htb]
  \caption{Reduction in Gaussian Case: Distinguish between $\sigma_1$ and
    $\sigma_2$ via Lindbladian Simulation}%
  \label{alg:reduction-lifting-Gaussian}
\begin{algorithmic}[1]
\Require Query access to a block-encoding of $H$, which is either $H_1=\sigma_1$ or $H_2=\sigma_2$.
\Ensure Decide the input being either $\sigma_1$ or $\sigma_2$.
\State{}$\rho_0 \gets \ket{+}\bra{+}$, $t \gets 9c^2$, $\texttt{flag} \gets \texttt{false}$.
\State{}Let $\mathcal{L}$ denote the Lindbladian with the single jump  $H$.
\For{$i = 1$ to $100$}
    \State{} Prepare $\rho_0$.
    \State{} Perform Lindbladian simulation $\rho_t \gets e^{t\mathcal{L}}(\rho_0)$, and measure $\rho_t$ in $\{\ket{+},\ket{-}\}$ basis.
    \State{}If the measurement outcome is $\ket{-}$, set $\texttt{flag} \gets \texttt{true}$.
\EndFor{}
\State{}If $\texttt{flag} = \texttt{false}$ then \textbf{return} $\sigma_1$; else  \textbf{return} $\sigma_2$.
\end{algorithmic}
\end{algorithm}

Similarly, we show the lower bounds for stable twirling and compound twirling
procedure, yielding the following results.
\begin{theorem}%
  \label{thm:algorithm-Lindbladian-complexity-general-lower-bounds}
  Our Hamiltonian twirling Lindbladian simulation algorithm achieves optimal time complexity for Lindbladians in~\Cref{tab:Hamiltonian-twirling-collection-results}. 
\end{theorem}

Specifically, for the Hamiltonian twirling channel induced by Poisson
distributions, we prove an $\Omega(t)$ bound, which gives the following
corollary.
\begin{corollary}%
  \label{thm:Lindbladian-comparison}
  The optimal (worst-case) complexity for simulating a Lindbladian with a single
  unitary jump operator is $\Theta\rbra{t}$.
  The optimal (worst-case) complexity for simulating a Lindbladian with a single
  Hermitian jump operator is $\Theta\rbra{\sqrt{t}}$,
\end{corollary}

This also indicates that the Hermitian condition plays a key role in achieving
the $\widetilde{O}(\sqrt{t})$ fast-forwarding as in~\cite{SGR+25} and this is not generally
true for a Lindbladian with a single non-Hermitian jump operator.

We note that, in contrast to the no–fast-forwarding result of~\cite{CL17}, which proves an $\Omega(t)$ lower bound for Lindbladians with a single jump operator, our setting differs in two essential aspects:~\cite{CL17} assumes sparse access to the jump operator and considers worst-case instances whose dimension grows with time, whereas we assume block-encoding access and obtain worst-case instances of constant dimension.

\section{Summary and Discussions}

We establish a systematic connection between Markovian open-system dynamics and
Hamiltonian twirling channels.
Leveraging the L\'{e}vy-Khintchine theorem, we characterize when dissipative
evolutions can be represented as averages over real-time Hamiltonian dynamics,
thereby providing a unified description of broad classes of Lindbladians.
For Gaussian, compound Poisson, and symmetric stable twirling distributions, we
show that the resulting dynamics can be simulated with Hamiltonian simulation
time $\Theta(t^{1/\alpha})$ for $1 \le \alpha \le 2$.
This yields explicit upper bounds and identifies a large family of Lindbladians
that admit polynomial fast-forwarding using only closed-system Hamiltonian
simulation, without ancillas or additional gate overhead, highlighting both
theoretical significance and practical relevance.
Finally, we prove that these complexities are asymptotically optimal via a novel
information-theoretic argument, establishing fundamental lower bounds on
fast-forwarding quantum simulations for the first time.

Our results point to several open questions of practical relevance.
On the theoretical side, it remains to be understood which classes of
Lindbladian dynamics admit fast-forwarded simulation, and whether superquadratic
or even exponential speedups are achievable in broader settings.
Clarifying the conditions for fast-forwarding, and its possible connections to
quantum metrology as in the Hamiltonian case~\cite{AA17}, would further inform
the design of efficient algorithms.
In addition, motivated by recent advances in the simulation of non-Markovian
open quantum systems~\cite{LW23a,LLW+24}, it is natural to ask whether
Hamiltonian twirling techniques can be extended beyond the Markovian regime.
From an algorithmic perspective, Lindbladian simulation underpins a wide range
of quantum computational tasks~\cite{VWI09}, including Gibbs and ground-state
preparation~\cite{KB16,DCL24,CKBG25,DLL25,ZDH+25}, optimization and
control~\cite{HLL+24,CLW+25}, and solving differential equations~\cite{SGAZ25}.
Using the method of~\cite{SGR+25}, our algorithm
(\Cref{alg:gaussian-twirl-cutoff}) already enables Gibbs state preparation.
A key direction for future work is to determine how far this approach can be
extended to accelerate other applications of Lindbladian simulation.

\bibliography{main.bib}

\newpage

\appendix
\onecolumngrid
\newpage

\tableofcontents
\clearpage

\section{Preliminaries}

\subsection{Some Results about Quantum Channels}%
\label{subsec:prelim-quantum-channel}

\subsubsection{The Lindbladian Master Equation}%
\label{subsec:Lindbladian-master-equation}

We review the Lindbladian master equation, which characterizes the dynamics of Markovian open quantum systems. Here we present the definition in the Schr\"{o}dinger picture.

\begin{definition}[Quantum Markovian Semigroup]
  Let $\mathcal{H}$ denote a finite-dimensional Hilbert space.
  A family of quantum channels $\{\Phi_t\}$ on $\mathcal{D}\rbra{\mathcal{H}}$
  forms a quantum Markovian semigroup if
\begin{itemize}
  \item $\Phi_0 = \mathrm{Id}$,
  \item $\Phi_{t+s} = \Phi_t \circ \Phi_s$ for all $t,s \ge 0$,
  \item each $\Phi_t$ is completely positive and trace-preserving (CPTP),
  \item $\Phi_t$ is continuous in $t$ in operator norm.
\end{itemize}
\end{definition}

Then, by~\cite{GKS76, Lin76}, the generator $\mathcal{L}$ of the semigroup, defined by
\begin{equation*}
  \mathcal{L}\sbra{\rho} = \lim_{t \to 0^+} \frac{\Phi_t\rbra{\rho} - \rho}{t},
\end{equation*}
in an $N$-dimensional Hilbert space with $\{\sigma_j\}_{j=1}^{N-1}$ being a basis of traceless operators,
has the following 
\emph{Gorini-Kossakowski-Sudarshan (GKS)} form
\begin{equation}%
    \label{eq:GKS-form-Lindbladian}
     \frac{d\rho}{dt} = \mathcal{L}[\rho]
  = -\ii \sbra*{H, \rho} + \sum_{j,k=1}^{N^2-1} A_{jk}\rbra*{\sbra{\sigma_j \rho, \sigma_k^{\dagger}} + \sbra{\sigma_j, \rho\sigma_k^{\dagger}}},
\end{equation}
where $H$ is a self-adjoint operator, and $A$ is a positive semidefinite $(N^2-1)\times (N^2-1)$ matrix with entries $A_{jk}$.

A more familiar form is the
\emph{Lindbladian} form:
\begin{equation}%
  \label{eq:lindblad}
  \frac{d\rho}{dt} = \mathcal{L}[\rho]
  = -\ii \sbra*{H, \rho} + \sum_{j} \rbra*{ L_j \rho L_j^\dagger
  - \frac{1}{2}\cbra*{L_j^\dagger L_j, \rho }},
\end{equation}
where $H = H^\dagger$ is the \emph{effective (system) Hamiltonian} and $\{L_j\}s$ are the
\emph{jump operators} characterizing the system-environment interaction.

\subsubsection{Schur Channels}
We first recall the definition of Schur channel.
\begin{definition}[Schur Maps and Schur Channels, see~{\cite[Definition 4.16]{Wat18}}]%
  \label{def:Schur-channel}
  Let $\mathcal{H}$ be a $d$-dimensional Hilbert space.
  A map $\Phi$ from the set of linear operators on $\mathcal{H}$ to the set of linear
  operators on $\mathcal{H}$ is said to be a Schur map, if there exists a linear
  operator $A$ such that
  \begin{equation*}
    \Phi(X) = A\odot X,
  \end{equation*}
  where $\odot$ denotes the Hadamard product (also known as the entrywise product).
  If $\Phi$ is a channel, then we say it is a Schur channel.
\end{definition}

\begin{proposition}[Completely Positive and Trace-Preserving Criteria of Schur Maps,~{\cite[Proposition 4.17 and 4.18]{Wat18}}]%
  \label{prop:cptp-Schur-channel}
  Let $\mathcal{H}$ be a $d$-dimensional Hilbert space, and $A$ be a linear
  operator on $\mathcal{H}$.
  Consider the Schur map $\Phi$ defined as $\Phi(X) = A\odot X$, then $\Phi$ is completely
  positive if and only if $A$ is positive semidefinite, and trace-preserving if
  and only if $\bra{j}A\ket{j} = 1$ for all $0\le j\le d-1$.
\end{proposition}

\subsection{The Vectorization Technique}%
\label{subsec:vectorization}

In this part, we review our basic tool, namely the operator-vector
correspondence, or the vectorization technique discussed in~\cite{Wat18}.

Let $\mathcal{X}$ and $\mathcal{Y}$ be two complex Euclidean spaces.
There is a correspondence
$\vecop: \mathrm{L}(\mathcal{Y},\mathcal{X})\to \mathcal{X}\otimes \mathcal{Y}$
given by
\begin{equation*}
  \vecop\rbra*{E_{a,b}} = e_a\otimes e_b,
\end{equation*}
or in Dirac notation,
\begin{equation*}
  \vecop\rbra*{\ket{i}\bra{j}} = \ket{i}\ket{j}.
\end{equation*}

Note that the $\vecop$ operation depends on the choice of the basis. We have the
following crucial property for this operation.

\begin{proposition}[{\cite[Equation 1.132]{Wat18}}]
  For any $A_0\in \mathrm{L}(\mathcal{X}_0,\mathcal{Y}_0)$,
  $A_1\in \mathrm{L}(\mathcal{X}_1,\mathcal{Y}_1)$, and
  $B\in \mathrm{L}(\mathcal{X}_1,\mathcal{X}_0)$, we have
  \begin{equation*}
    \rbra*{A_0\otimes A_1} \vecop\rbra*{B} = \vecop\rbra*{A_0BA_1^{\intercal}}.
  \end{equation*}
\end{proposition}

\subsection{Fourier Transform}

Recall that for a function $f$, its Fourier transform $\mathcal{F}\rbra{f}$ is
defined as
\begin{equation*}
  \mathcal{F}\rbra*{f}\rbra{\xi} = \int_{\mathbb{R}^d} f(x) \exp\rbra*{-2\pi \ii \xi\cdot x}\dd{x}.
\end{equation*}

We need the following proposition stating that the Fourier transform of a
Gaussian distribution density function is also Gaussian.
\begin{proposition}[See~{\cite{Fol99}}]%
  \label{prop:Gaussian-Fourier-transform}
  Let $f(x): \mathbb{R}^d\to \mathbb{R}$ be the function such that
  $f(x) = \exp\rbra{-\pi a\norm{x}^2}$.
  Then, the Fourier transform of $f$ is given as
  $\mathcal{F}\rbra{f}(\xi) = a^{-d/2}\exp\rbra{-\pi \norm{\xi}^2/a }$.
\end{proposition}

The above proposition is called Hubbard-Stratonovich transformation (\cite{Str57,Hub59})
and applied as a way of simulating imaginary-time evolution using real-time evolutions
(see~\cite{CS17, ACNR22, GANW25}).

\subsection{Some Results of Probability Theory}%
\label{subsec:prob-preliminary}

We first recall some important definitions in probability theory.

\begin{definition}[Characteristic Function, see~{\cite[Definition 2.1]{Sat13}}]%
  \label{def:characteristic-function}
  Let $\mu$ be a probability measure on $\mathbb{R}^d$. Its characteristic function $\hat{\mu}$ is defined as
  \begin{equation*}
    \hat{\mu}\rbra*{z} = \int_{\mathbb{R}^d} \e^{\ii z \cdot x } \mu\rbra*{\dd{x}}, z\in \mathbb{R}^d.
  \end{equation*}
  The characteristic function of a random variable $X$ is given by $\E\sbra{\e^{\ii z\cdot X}}$.
\end{definition}

\begin{definition}[Convolution of Two Distributions, see~{\cite[Definition 2.4]{Sat13}}]
  Let $\mu_1$ and $\mu_2$ be two probability measures on $\mathbb{R}^d$. Their convolution $\mu = \mu_1\star \mu_2$
  is a probability measure given by
  \begin{equation*}
    \mu \rbra*{B} = \iint_{\mathbb{R}^d\times \mathbb{R}^d} 1_{B}\rbra*{x+y} \mu_1\rbra*{\dd{x}} \mu_2 \rbra*{\dd{y}}, B\in \mathcal{B}\rbra*{\mathbb{R}^d},
  \end{equation*}
  where $\mathcal{B}\rbra*{\mathbb{R}^d}$ denote the Borel measurable sets in $\mathbb{R}^d$.
\end{definition}

Here are some basic properties of the characteristic function and the
convolution: if $X_1$ and $X_2$ are two independent random variables with
distributions $\mu_1$ and $\mu_2$, then $X_1+X_2$ has distribution $\mu_1\star \mu_2$.
We also have, for $\mu = \mu_1\star \mu_2$, $\hat{\mu} = \hat{\mu_1} \hat{\mu_2}$.

\subsubsection{Infinitely Divisable Distributions}
In the following, we recall the concept of infinitely divisible distribution.
Intuitively speaking, if a random variable $X$ follows an infinitely divisible
distribution, then for any positive integer $n$, it can be written as a sum of
$n$ identical and independent random variables.

\begin{definition}[Infinitely Divisible Distributions, see~{\cite[Definition 7.1]{Sat13}}]%
  \label{def:infinite-divisible}
  A probability measure $\mu$ on $\mathbb{R}^d$ is said to be infinitely divisible,
  if for every positive integer $n$, there exists a probability measure $\mu_n$ on $\mathbb{R}^d$, such that
  the $n$-fold convolution, $\mu_n\star \mu_n\star \dots \star \mu_n$ equals  $\mu$.
\end{definition}

\begin{table*}[htb]
\centering
\begin{ruledtabular}
\begin{tabular}{cccc}
\textbf{Distribution} &
\textbf{PMF / PDF} &
\textbf{Characteristic Function} &
\textbf{Infinitely Divisible?} \\
\midrule

Bernoulli$(p)$ &
$p^x(1-p)^{1-x}$, $x\!\in\!\{0,1\}$ &
$(1-p)+p\,\mathrm{e}^{\,\mathrm{i} t}$ &
No \\
\midrule

Uniform$[a,b]$ &
$\frac{1}{b-a}$ on $[a,b]$ &
$\frac{\mathrm{e}^{\mathrm{i} t b}-\mathrm{e}^{\mathrm{i} t a}}{\mathrm{i}t(b-a)}$ &
No \\
\midrule

Gaussian $N(0,\sigma^2)$ &
$\frac{1}{\sqrt{2\pi\sigma^2}}
 \mathrm{e}^{-\frac{x^2}{2\sigma^2}}$ &
$\exp\rbra*{- \frac{1}{2} \sigma^2 t^2}$ &
Yes \\
\midrule

Poisson$(\lambda)$ &
$\frac{ \mathrm{e}^{-\lambda}\lambda^k }{k!}$ &
$\exp\rbra*{\lambda(\mathrm{e}^{\,\mathrm{i} t}-1)}$ &
Yes \\
\midrule

Laplace$(0,b)$ &
$\frac{1}{2b}\,\mathrm{e}^{-\frac{\abs{x}}{b}}$ &
$\frac{1}{1+b^2 t^2}$ &
Yes \\
\midrule

Symmetric Cauchy$(\gamma)$  &
$\frac{\gamma}{\pi(\gamma^2+x^2)}$ &
$\exp\rbra*{-\gamma |t|}$ &
Yes \\
\midrule

Symmetric $\alpha$-stable &
(no closed form) &
$\exp\rbra*{-c |t|^\alpha}$ &
Yes \\
\end{tabular}
\end{ruledtabular}
\caption{Common distributions, their probability mass (or density) functions,
  characteristic functions, and infinite divisibility.}
\end{table*}

The following L\'evy-Khintchine representation theorem characterizes when a
distribution is infinitely divisible via its characteristic function.

\begin{theorem}[L\'evy-Khintchine Representation Theorem, see~{\cite[Theorem 8.1]{Sat13}}]%
  \label{thm:levy-khintchine}
  Let $D$ denote the set $\cbra*{\norm{x}\le 1}$.
  \begin{itemize}
    \item If $\mu$ is an infinitely divisible probability distribution on $\mathbb{R}^d$, then its characteristic function satisfies
          \begin{equation}%
            \label{eqn:levy-khintchine-form}
            \hat{\mu}\rbra*{z}
            = \exp \sbra*{-\frac{1}{2}\ave{z, Az} + \ii \ave{\gamma, z}
              + \int_{\mathbb{R}^d} \rbra*{ \e^{\ii z\cdot x} -1 - \ii z\cdot x 1_{D}\rbra*{x} }\nu\rbra*{\dd{x}} },
          \end{equation}
          where $A$ is a symmetric positive semidefinite $d\times d$ matrix, $\nu$ is a Borel measure on $\mathbb{R}^d$ that satisfies
          \begin{equation}%
            \label{eqn:levy-measure-condition}
            \int_{\mathbb{R}^d} \min\cbra*{\norm{x}^2, 1} \nu \rbra*{\dd{x}} < +\infty, \nu \rbra{\cbra{0}} = 0,
          \end{equation}
          and $\gamma\in \mathbb{R}^d$.
          Moreover, the representation of $\hat{\mu}$ by $A$, $\nu$, and $\gamma$ is
          unique.
    \item Conversely, for a positive semidefinite matrix $A$, a measure $\nu$
          satisfying~\Cref{eqn:levy-measure-condition}, and $\gamma\in \mathbb{R}^d$,
          there exists a distribution $\mu$ whose characteristic function
          satisfies~\Cref{eqn:levy-khintchine-form}.
  \end{itemize}
\end{theorem}

\begin{remark}
  As noted in~\cite[Section 3.8]{Var01} (see also~{\cite[Remark 8.4]{Sat13}}),
  the function $x1_D\rbra{x}$ could be replaced by other bounded continuous
  functions $\theta\rbra{x}$ satisfying
  $\norm{\theta\rbra{x}-x}\le C\norm{x}^{3}$ for some constant $C$.
  One such example is $\theta\rbra{x}=\frac{x}{1+\norm{x}^{2}}$.
\end{remark}

As noted in~{\cite[Example 7.2]{Sat13}}, Gaussian, Cauchy and
$\delta$-distributions on $\mathbb{R}^d$ are infinitely divisible.
Poisson, geometric, negative binomial, exponential, and $\Gamma$ distributions
on $\mathbb{R}$ are also infinitely divisible.
Besides these, another important class of infinitely divisible distributions is
compound Poisson, which we will discuss later.

\subsubsection{Gaussian Distributions}

We use $\mathcal{N}(\mu, \sigma^2)$ to denote the normal distribution with mean $\mu$ and
variance $\sigma^2$ throughout the paper.

\begin{lemma}[Mill's Inequality, See~\cite{Hu17}]
  Suppose a random variable $X$ follows the normal distribution $\mathcal{N}(0,\sigma^2)$.
  Then, for every $t > 0$, we have
  \begin{equation*}
    \Pr\rbra*{\abs{X} > t} \le \sqrt{\frac{2}{\pi}} \frac{\sigma}{t} \exp \rbra*{-\frac{t^2}{2\sigma^2}}.
  \end{equation*}
\end{lemma}

\subsubsection{Stable Distributions}\label{appendix:stable-distribution}
We first recall some basic knowledge of the symmetric stable distributions.
For more details, we refer readers to the survey~\cite{Jan22} and the book~\cite{Zol86}.
\begin{definition}[Stable Distributions, see~{\cite[Definition 1, Chapter
    VI]{Fel09}} and~{\cite[Definition 13.1]{Sat13}}]%
  \label{def:stable-distribution}
  Let $\mu$ be an infinitely divisible measure on $\mathbb{R}^d$ and independent
  random variables $X,X_1,X_2\sim \mu$.
  Then, we say $\mu$ is stable, if for any $A>0$ and $B>0$, there exists $C >0$
  and a real number $D$, such that
  \begin{equation*}
    AX_1+BX_2 \overset{d}{=} CX + D,
  \end{equation*}
  where $\overset{d}{=}$ means equality in distribution, i.e., the two random variables have the same cumulative distribution function.
\end{definition}

\begin{theorem}[Stable Index, see~{\cite[Theorem 1, Chapter VI]{Fel09}}]%
  \label{thm:stable-index}
  For any stable random variable $X$ following a distribution $\mu$, there
  exists a constant $\alpha\in \interval[open left]{0}{2}$, such that $A, B, C$
  in~\Cref{def:stable-distribution} satisfies
  \begin{equation*}
    C^{\alpha} = A^{\alpha}+B^{\alpha}.
  \end{equation*}
  $\alpha$ is called the stable index of the distribution $\mu$, and we also say $\mu$ is $\alpha$-stable.
\end{theorem}

For example, the normal distribution $\mathcal{N}\rbra{0, \sigma^2}$ is $2$-stable.
Actually, it can be shown that (see~\cite[Theorem 14.1]{Sat13}) a distribution
is $2$-stable if and only if it is Gaussian.

For general $\alpha$-stable distributions, it is hard to explicitly write out their
probability density functions.
However, their characteristic functions have explicit forms.
For simplicity, we consider only the symmetric cases.

\begin{theorem}[Characteristic Functions of Stable Distributions,
  see~{\cite[Theorem 14.4]{Sat13}}]%
  \label{thm:characteristic-function-stable}
  Let $\mu$ be a probability measure on $\mathbb{R}^d$.
  Then $\mu$ is rotation-invariant and $\alpha$-stable for $0 < \alpha \le 2$ if and only if
  \begin{equation*}
    \hat{\mu}\rbra{z} = \exp \rbra*{-c\norm{z}^{\alpha}}
  \end{equation*}
  for some $c >0$.
\end{theorem}

For simplicity, we use $\StableS{\alpha}{c}$ to denote the symmetric $\alpha$-stable
distribution with parameter $c$ in its characteristic function.

\begin{theorem}[Absolute Moment for Symmetric Stable Distributions,
  see~{\cite[Proposition 2.5]{Zol86}}]%
  \label{thm:general-absolute-moment}
  Let $X$ be a random variable that follows $\StableS{\alpha}{\lambda}$. For any $0 < s < \alpha $,
  we have
  \begin{equation*}
    \E\sbra*{\abs{X}^s} \le 2 \lambda^{s/\alpha} \sin \rbra*{\frac{\pi s}{2}} \Gamma\rbra*{s}\Gamma\rbra*{1-\frac{s}{\alpha}}.
  \end{equation*}
  Specifically, we have
  \begin{equation*}
     \E\sbra*{\abs{X}} \le  C_{\alpha} \lambda^{1/\alpha}
   \end{equation*}
   for some constant $C_{\alpha}$ that depends only on $\alpha$.
\end{theorem}

\subsubsection{Compound Poisson Distributions}
\begin{definition}[Compound Poisson Distribution, see~{\cite[Definition 4.1]{Sat13}}]
  A distribution $\mu$ on $\mathbb{R}^d$ is compound Poisson, if there exists
  some $\lambda > 0$ and some distribution $\sigma$ on $\mathbb{R}^d$ such that
  $\sigma\rbra{\cbra{0}}=0$, and
  \begin{equation*}
    \hat{\mu}\rbra*{z} = \exp \rbra*{\lambda \rbra*{\hat{\sigma}\rbra{z}-1}}.
  \end{equation*}
\end{definition}

An equivalent definition of a compound Poisson is as follows (see~\cite[Chapter
XII.2]{Fel76}). Suppose $N\sim \Poisson\rbra{\lambda}$. Let
$X_1, X_2, \dots, X_n, \dots$ be a sequence of independent and identically
distributed random variables. The distribution of the random variable
$S_{N} = \sum_{j=1}^{N} X_{j}$ is a compound Poisson. When $X_{j}$s are all $1$,
the distribution reduces to the standard Poisson distribution $\Poisson\rbra{\lambda}$.

\section{Hamiltonian Twirling Channels}%
\label{sec:twirling}

For convenience, we introduce the following definition of Hamiltonian Twirling
channels. We note that, in~\cite{AKMV10}, the authors have already proposed a
slightly more general and abstract notion called \emph{twirling superoperators}
for a locally compact, second countable, Hausdorff (topological space) Lie group
with a projective representation and a probability measure.

\begin{definition}[Hamiltonian Twirling Channel, see~{\cite[Proposition 4.2]{AKMV10}}]%
  \label{def:Hamiltonian-twirling-channel}
  Let $\mu$ denote a probability distribution on $\mathbb{R}$ and
  $p_{\mu}$ denote its probability density function.
  Let $H$ be an $n$-qubit Hamiltonian.
  We call the following channel, which acts on $n$-qubit density operators,
  \begin{equation*}
    \Phi_{H,\mu}\rbra{\rho}
    = \int_{-\infty}^{\infty}  p_{\mu}\rbra{s} \e^{\ii H s} \rho  \e^{-\ii H s}\dd{s}
    = \E_{s\sim \mu}\sbra{ \e^{\ii H s} \rho  \e^{-\ii H s}},
  \end{equation*}
  the Hamiltonian twirling channel of $H$ with distribution $\mu$.
  If $H$ is clear from the context, we will simply write $\Phi_{\mu}$.
\end{definition}

\begin{proposition}[Every Hamiltonian Twirling Channel is a Schur Channel]%
  \label{prop:hamiltonian-twirling-is-schur}
  Let $H$ be an $n$-qubit Hamiltonian with eigendecomposition
  $H =\sum_j \lambda_j \ket{j}\bra{j}$, $\rho$ be an $n$-qubit density operator,
  $\mu$ be a fixed probability distribution on $\mathbb{R}$ with density
  function $p_{\mu}$.
  Let $\rho_{jk} = \bra{j} \rho \ket{k}$.
  Then, we have
  \begin{equation*}
    \Phi_{\mu} \rbra*{\rho} = \sum_{jk} \hat{\mu}\rbra*{\lambda_j-\lambda_k} \rho_{jk} \ket{j}\bra{k},
 \end{equation*}
 where $\hat{\mu}$ is the characteristic function of the distribution
 $\mu$, i.e.,
 \begin{equation*}
   \hat{\mu}\rbra{\omega} = \int_{\mathbb{R}} p_{\mu}\rbra{s} \e^{\ii \omega s} \dd{s}.
 \end{equation*}
 In other words, $\Phi_{\mu}$ is a Schur channel with the matrix with entry
 $jk$ being $\hat{\mu}\rbra*{\lambda_j-\lambda_k}$.
\end{proposition}

\begin{proof}
  This can be proved by a direct computation.
  In fact, note that
  $\e^{-\ii H s} = \sum_j \e^{-\ii \lambda_j s} \ket{j}\bra{j}$.
  Therefore, we have
  \begin{equation*}
    \begin{aligned}
      \Phi_{\mu}\rbra*{\rho}
      &= \int_{\mathbb{R}} p_{\mu}\rbra{s} \e^{\ii H s} \rho \e^{-\ii H s}  \dd{s} \\
      &= \int_{\mathbb{R}} p_{\mu}\rbra{s} \sum_{jk} \e^{\ii \rbra*{\lambda_j-\lambda_{k}} s}  \rho_{jk}\ket{j}\bra{k}   \dd{s} \\
      &= \sum_{jk}\rbra*{ \int_{\mathbb{R}} p_{\mu}\rbra{s} \e^{\ii \rbra*{\lambda_j-\lambda_{k}} s}  \dd{s}} \; \rho_{jk}\ket{j}\bra{k},   \\
    \end{aligned}
 \end{equation*}
 which completes the proof.
\end{proof}

\subsection{Lindbladian Generator for Hamiltonian Twirling Channel with
  Infinitely Divisible Distributions}%
\label{subsec:generator-hamiltonian-twirling-infinite-divisible}

One central result about twirling channels (or twirling superoperators) is the
correspondence theorem about the conditions of the probability measures required
for a twirling channel being Markovian. We will state the theorem for a single
measure, that if it is infinitely divisible, then the resulting Hamiltonian
twirling channel can be written as the dynamics of a Lindbladian. This form of
Lindbladian, to the best of our knowledge, first appears in~\cite{Hol95}
(see~{\cite[Theorem 1]{Hol95}}) as non-commutative L\'evy-Khintchine type formulae
to characterize translation-covariant Lindbladians. Our statement of the theorem
here corresponds to a more abstract result in~\cite{AKMV10}, where the authors
considered a semigroup of measures on a Lie group and a unitary representation of it with L\'evy-Khintchine
representations in GKS form, while we consider only the special case that the Lie group is $(\mathbb{R},+)$ and the unitary representation corresponds to $\mathrm{e}^{\mathrm{i}H\cdot }$ in explicit Lindbladian form.

\begin{theorem}[L\'evy-Lindbladian Correspondence of a Hamiltonian Twirling
  Channel, see~{\cite[Theorem 5.1]{AKMV10}}]%
  \label{thm:levy-lindbladian-correspondence-hamiltonian-twirling}
  Let $H$ be an $n$-qubit Hamiltonian, $\mu$ denote an infinitely
  divisible probability distribution on $\mathbb{R}$, and
  $\Phi_{H,\mu}\rbra{\cdot}$ denote the Hamiltonian twirling channel of $H$ with
  distribution $\mu$.
  Then, $\Phi_{H,\mu}$ can be written as $\e^{\mathcal{L}}$ with
  $\mathcal{L}$ being a time-independent Lindbladian operator, whose effective
  Hamiltonian and jump operators commute with $H$.
\end{theorem}

\begin{proof}
  Consider the spectral decomposition of the Hamiltonian as $H = \sum_j \lambda_j\ket{j}\bra{j}$.
  By~\Cref{prop:hamiltonian-twirling-is-schur}, we know that for any density
  operator $\rho$,
  $\Phi_{\mu}\rbra{\rho} = \sum_{jk} \hat{\mu}\rbra*{\lambda_j-\lambda_k} \rho_{jk} \ket{j}\bra{k}$,
  where $\hat{\mu}$ is the characteristic function of $\mu$.
  Since $\mu$ is infinitely divisible, by~\Cref{thm:levy-khintchine}, we know
  we can write $\hat{\mu}$ as $\e^{\psi}$ for some function $\psi$ of the form
  \begin{equation}
    \label{eqn:levy-form-character-in-proof-correspondence}
           \psi\rbra*{z}
           = -\frac{\sigma^2}{2}z^2 + \ii \gamma z
             + \int_{\mathbb{R}} \rbra*{ \e^{\ii z s} -1 - \ii z s 1_{D}\rbra*{s} }\nu \rbra*{\dd{s}},
  \end{equation}
  for some $\sigma \ge 0$, $\gamma\in \mathbb{R}$, and a L\'evy measure $\nu$, and $D = \{x: \abs{x}\le 1\}$.
  Then, it is direct to define the Lindbladian $\mathcal{L}$ via the equation
  \begin{equation*}
    \mathcal{L}\rbra*{\rho} =\sum_{j,k} \psi\rbra*{\lambda_j-\lambda_k} \rho_{jk} \ket{j}\bra{k}
  \end{equation*}
  by analyzing each term of $\psi\rbra{z}$
  in~\Cref{eqn:levy-form-character-in-proof-correspondence}.

  Noting that
  \begin{equation*}
    \sbra{H, \ket{j}\bra{k}} = \rbra{ \lambda_j-\lambda_k} \ket{j}\bra{k},
  \end{equation*}
  we know
  the term $\ii \gamma z$ in $\psi\rbra{z}$ corresponds to
  $\ii \gamma \sbra{H, \cdot}$ in the Lindbladian.
  By the same reasoning,
  the term $-\ii z \int_{\mathbb{R}} 1_D\rbra{s} \nu \rbra{\dd{s}} $ in $\psi\rbra{z}$
  corresponds to $-\ii \sbra{H, \cdot} \int_{\mathbb{R}} 1_D\rbra{s} \nu \rbra{\dd{s}}$ in the Lindbladian.

  Similarly, since
  \begin{equation*}
  H \ket{j}\bra{k} H - \frac{1}{2} \cbra*{\ket{j}\bra{k}, H^2 } = \rbra*{ \lambda_j  \lambda_k - \frac{\lambda_j^2+\lambda_k^2}{2}} \ket{j}\bra{k} = -\frac{\rbra{\lambda_j-\lambda_k}^2}{2}\ket{j}\bra{k},
  \end{equation*}
  the term $-\frac{\sigma^2}{2} z^2$ in $\psi\rbra{z}$ corresponds to
  a jump operator $H$ in the Lindbladian.

  For the term $\e^{\ii z s} - 1$ in $\psi\rbra{z}$, noting that
  \begin{equation*}
    \e^{\ii H s} \ket{j}\bra{k} \e^{-\ii H s}
    - \frac{1}{2} \cbra*{\ket{j}\bra{k}, \e^{-\ii H s} \e^{\ii H s}}
    = \rbra*{\e^{\ii \rbra{\lambda_j-\lambda_k}s} - 1} \ket{j}\bra{k},
  \end{equation*}
  we know it corresponds to a jump operator $ \e^{\ii H s}$ in the Lindbladian.

  To summarize, we could write the Lindbladian operator as
  \begin{equation*}
    \mathcal{L}\rbra*{\rho}
    = \ii \gamma \sbra*{H, \rho} + \sigma^2 H \rho H - \frac{\sigma^2}{2}  \cbra*{\rho,H^2}
      + \int_{\mathbb{R}} \rbra*{ \e^{\ii H s} \rho \e^{-\ii Hs} -1 - \ii s \sbra*{H,\rho}  1_{D}\rbra*{s} }\nu\rbra*{\dd{s}}.
    \end{equation*}
    Note that this matches the standard form of Lindbladians
    \begin{equation*}
      \mathcal{L}\rbra*{\rho} = - \ii \sbra*{H_{\text{eff}},\rho} + \sum_j L_j \rho L_j^{\dagger} - \frac{1}{2} \cbra*{L_j^{\dagger}L_j, \rho}
    \end{equation*}
    with the effective Hamiltonian $H_{\text{eff}} =  \rbra*{\int_{-1}^{1} s \; \nu \rbra*{\dd{s}}-\gamma} H$, and jump operators $L_0 = \sigma H$, $L_{s} =  \e^{\ii H s} \sqrt{\frac{\dd{\nu}}{\dd{t}}\rbra{s}}$ for $s\in \mathbb{R}$.
\end{proof}

In addition, we establish the following converse theorem, showing that if Hamiltonian twirling channels form a Markovian semigroup for all Hamiltonians, then the associated distributions are infinitely divisible.
\begin{theorem}
  Let $\{ \mu_s \}_{s\ge 0}$ denote a family of probability distributions. Suppose for any Hamiltonian $H$, the Hamiltonian twirling channels
  \begin{equation*}
    \Phi_{s} \rbra{\cdot}= \int_{\mathbb{R}} \e^{\ii Ht} \rbra{\cdot} \e^{-\ii Ht} \mu_{s}\rbra{\dd{t}}
  \end{equation*}
  satisfy $\Phi_s\circ \Phi_t = \Phi_{s+t}$ for any $s,t\ge 0$. 
  Then, $\{\mu_t\}_{t\ge 0}$ forms a semigroup under the convolution operation.
  Specifically, for any $t> 0$, $\mu_{t}$ is infinitely divisible.
\end{theorem}

\begin{proof}
  It suffices to show $\hat{\mu}_s(x) \hat{\mu}_t(x) = \hat{\mu}_{s+t}(x)$ for any $s,t\ge 0$ and $x\in \mathbb{R}$.
  Let $H = x \ket{0}\bra{0}$ and 
  \begin{equation*}
    \rho = \frac{1}{2}
    \begin{pmatrix}
      1 & 1 \\
      1 & 1 \\
    \end{pmatrix}
  \end{equation*}
  in basis $\ket{0}$ and $\ket{1}$.
  By~\Cref{prop:hamiltonian-twirling-is-schur}, we know 
  \begin{equation*}
    \Phi_t(\rho) = 
    \frac{1}{2}
    \begin{pmatrix}
      \hat{\mu}_t(0) & \hat{\mu}_t(x) \\
      \hat{\mu}_t(-x) & \hat{\mu}_t(0) \\
    \end{pmatrix},
  \end{equation*}
  and   \begin{equation*}
    \Phi_s(\Phi_t(\rho)) = 
    \frac{1}{2}
    \begin{pmatrix}
      \hat{\mu}_s(0) \hat{\mu}_t(0) & \hat{\mu}_s(x) \hat{\mu}_t(x) \\
      \hat{\mu}_s(-x) \hat{\mu}_t(-x) & \hat{\mu}_s(0)\hat{\mu}_t(0) \\
    \end{pmatrix}.
  \end{equation*}
  Similarly, we get 
    \begin{equation*}
    \Phi_{s+t}(\rho) = 
    \frac{1}{2}
    \begin{pmatrix}
      \hat{\mu}_{s+t}(0) & \hat{\mu}_{s+t}(x) \\
      \hat{\mu}_{s+t}(-x) & \hat{\mu}_{s+t}(0) \\
    \end{pmatrix}.
  \end{equation*}
  Comparing the entries, we get $\hat{\mu}_{s+t}(x)= \hat{\mu}_s(x) \hat{\mu}_t(x)$ as we desired.
\end{proof}

\section{Lindbladian Simulation by Hamiltonian Twirling Channels}

\subsection{Gaussian Twirling}
We first prove a lemma for bounding the error for truncating Gaussian distributions.
\begin{lemma}%
  \label{lem:truncated-Gaussian}
  For $t , S > 0$, let $\mu_{t,S}$ denote the distribution with density
  function
  \begin{equation*}
q_{t,S}(s) \propto
\begin{cases}
\e^{-s^2/(2t)}, & \text{if } s \in \interval{-S}{S}. \\
0, & \text{otherwise}.
\end{cases}
\end{equation*}
Let $\mu_t$ denote the normal distribution $\mathcal{N}\rbra{0, t}$.
Then, the total variation distance of $\mu_{t,S}$ and $\mu$ is
bounded by $\sqrt{\frac{2t}{\pi}} \frac{1}{S}\exp\rbra*{-\frac{S^2}{2t}}$.
Specifically, for $S = \sqrt{2t\log\rbra*{\frac{4}{\varepsilon}}}$, we have
\begin{equation*}
  \TV\rbra*{\mu_{t,S}, \mu_t} \le \varepsilon/2.
\end{equation*}
\end{lemma}

\begin{proof}
  Set
  \begin{equation*}
    Z_{t,S} = \frac{1}{\sqrt{2\pi t}}\int_{-S}^S \e^{-s^2/(2t)} \text{d} s.
  \end{equation*}
  By Mill's inequality, we have
  $Z_{t,S}\ge 1 - \sqrt{\frac{2t}{\pi}} \frac{1}{S}\exp\rbra*{-\frac{S^2}{2t}}$.
  Denote $p_{t}(s)$ be the probability density function of $\mu_t$, by
  definition we know
  \begin{equation*}
    q_{t, S} = \frac{\e^{-s^2/(2t)} \mathbf{1}_{\interval{-S}{S}}}{Z_{t,S}\sqrt{2\pi t}} = \frac{p_{t}(s) \mathbf{1}_{\interval{-S}{S}}}{Z_{t,S}}.
  \end{equation*}
  Therefore,
  \begin{equation*}
    \begin{aligned}
      \TV\rbra*{\mu_t, \mu_{t,S}}
      &= \frac{1}{2} \int \abs{p_t(s)-q_{t,S}(s)} \textup{d} s \\
      &= \frac{1}{2}\rbra*{ \int_{-S}^{S} \abs{p_t(s)-q_{t,S}(s)} \textup{d} s + \int_{\abs{s}\ge S} \abs{p_t(s)-q_{t,S}(s)} \textup{d} s  } \\
      &= 1- Z \le \sqrt{\frac{2}{\pi}} \frac{\sqrt{t}}{S}\exp\rbra*{-\frac{S^2}{2t}}.
    \end{aligned}
  \end{equation*}
  Taking $S = \sqrt{2t\log\rbra*{\frac{4}{\varepsilon}}}$, we know
  \begin{equation*}
\sqrt{\frac{2}{\pi}} \frac{\sqrt{t}}{S}\exp\rbra*{-\frac{S^2}{2t}} \le \frac{1}{\sqrt{2}} \exp\rbra*{-\log\rbra*{\frac{4}{\varepsilon}}} \le \frac{\varepsilon}{2}.
  \end{equation*}
\end{proof}

We are now able to prove the following theorem.
\begin{theorem}%
  \label{thm:algorithm-Lindbladian-complexity}
  Simulating Lindbladians with a single Hermitian jump operator for evolution time $t$ to within diamond
  distance error $\varepsilon$ can be done 
  in $O(\sqrt{t\log (1/\varepsilon)})$
  Hamiltonian simulation time \emph{without any extra ancilla}.
\end{theorem}
\begin{proof}
  Let $\mathcal{E}$ denote the channel given by the algorithm.
  Using the notation in~\Cref{lem:truncated-Gaussian}, we can write
  \begin{equation*}
    \mathcal{E}\rbra*{\rho} = \int_{\abs{s}\le S} q_{t,S}\rbra*{s} \exp\rbra*{-\ii H s} \rho \exp\rbra*{\ii H s}.
  \end{equation*}
  Let $\mathcal{E}_{s}$ denote the quantum channel of applying the unitary
  $\exp\rbra*{-\ii H s}$, we get
  \begin{equation*}
    \mathcal{E} = \int_{\abs{s}\le S} q_{t,S}\rbra*{s} \mathcal{E}_s \textup{d} s.
  \end{equation*}
  By the Hamiltonian twirling property, we know
  \begin{equation*}
    \e^{t \mathcal{L}} = \int_{-\infty}^{\infty} p_{t}\rbra*{s} \mathcal{E}_s \textup{d} s.
  \end{equation*}
  Therefore, we have
  \begin{equation*}
    \begin{aligned}
      \norm*{\mathcal{E}-\e^{t \mathcal{L}}}_{\diamond}
      &= \norm*{  \int_{\abs{s}\le S} q_{t,S}\rbra*{s} \mathcal{E}_s \textup{d} s - \int_{-\infty}^{\infty} p_{t}\rbra*{s} \mathcal{E}_s \textup{d} s}_{\diamond} \\
      &= \norm*{  \int_{-\infty}^{\infty}\rbra*{ q_{t,S}\rbra*{s}- p_{t}\rbra*{s}} \mathcal{E}_s \textup{d} s}_{\diamond} \\
      &\le  \int_{-\infty}^{\infty} \abs*{q_{t,S}\rbra*{s}- p_{t}\rbra*{s} }\norm*{\mathcal{E}_s}_{\diamond} \textup{d} s \le  2 \TV\rbra*{\mu_{t}, \mu_{t,S}}\le \varepsilon
    \end{aligned}
  \end{equation*}
  by~\Cref{lem:truncated-Gaussian}.
\end{proof}

\subsection{Stable Twirling}%
\label{sec:stable-twirling-upper-bound}

\begin{proposition}
  Consider a Lindbladian $\mathcal{L}$ of the form
  \begin{equation*}
    \mathcal{L}\rbra{\rho} = C_{\alpha,c} \int_{\abs{s}>0} \dd{s} \frac{1}{\abs{s}^{1+\alpha}} \rbra*{\mathrm{e}^{-\mathrm{i}Hs}\rho \mathrm{e}^{\mathrm{i}Hs} - \rho},
  \end{equation*}
  where $C_{\alpha,c}$ is the normalization constant satisfying
  \begin{equation*}
    c \abs{x}^{\alpha} = C_{\alpha,c} \int_{\abs{s}>0}\dd{s} \rbra*{\cos \rbra*{xs}-1} \frac{1}{\abs{s}^{1+\alpha}},
  \end{equation*}
  and $H = \sum_j \lambda_j \ket{j}\bra{j}$ is a Hamiltonian.
  Then, we have
  \begin{equation*}
    \mathrm{e}^{t \mathcal{L}} \rbra*{\rho} = \sum_{jk} \mathrm{e}^{-t c \abs{\lambda_j-\lambda_k}^{\alpha}} \rho_{jk}\ket{j}\bra{k},
  \end{equation*}
  where $\rho_{jk} = \bra{j}\rho\ket{k}$.
\end{proposition}

\begin{theorem}[Fast Forwarding Lindbladian Simulation Upper Bound]%
  \label{thm:lindbladian-simulation-upper-bound-stable}
  Let $1 < \alpha \le 2$. For a Lindbladian $\mathcal{L}$ of the form
  \begin{equation*}
    \mathcal{L}\rbra{\rho} = C_{\alpha,c} \int_{\abs{s}>0} \dd{s} \frac{1}{\abs{s}^{1+\alpha}} \rbra*{\mathrm{e}^{-\mathrm{i}Hs}\rho \mathrm{e}^{\mathrm{i}Hs} - \rho},
  \end{equation*}
  where $C_{\alpha,c}$ is a normalization constant
  and $H = \sum_j \lambda_j \ket{j}\bra{j}$ is a Hamiltonian.
  or, equivalently
  \begin{equation*}
    \mathrm{e}^{t \mathcal{L}} \rbra*{\rho} = \sum_{jk} \mathrm{e}^{-t c \abs{\lambda_j-\lambda_k}^{\alpha}} \rho_{jk}\ket{j}\bra{k},
  \end{equation*}
  simulating the evolution of $\mathrm{e}^{t \mathcal{L}}$ can be done in $O\rbra{t^{1/\alpha}}$ Hamiltonian simulation time.
\end{theorem}

\begin{proof}
  Let $\mu_{t}$ denote the family of distributions $\StableS{\alpha}{ct}$.
  Then, the Hamiltonian twirling channel $\Phi_{\mu_t, H}$ implements $\mathrm{e}^{t \mathcal{L}}$.
  By~\Cref{thm:general-absolute-moment}, the expected simulation time is $O\rbra{t^{1/\alpha}}$.
\end{proof}

\subsection{Poisson and Compound Poisson Twirling}

In this part, we study the Hamiltonian twirling channel with a compound Poisson
distribution, which was also considered in~\cite{Vac05}.
Since a Poisson distribution is always compound Poisson, this also
generalizes the discussions in~\cite[Equation 6.26]{Kho91}.

\begin{theorem}[Compound Poisson Twirling]%
  \label{thm:compound-poisson-twirling}
  Let $\mu$ be a fixed distribution, and $\{\mu_{t}\}$ denote a set of compound Poisson distribution with its
  characteristic function
  $\hat{\mu}_{t} = \exp\rbra{t\rbra{\hat{\mu}-1}}$, where
  $\hat{\mu}$ is the characteristic function of
  $\mu$.
  Then, for a Hamiltonian $H$, the Hamiltonian twirling channel
  $\Phi_{H,\mu_t}$ can be written as $\e^{t \mathcal{L}}$ for
  the Lindbladian $\mathcal{L}$ satisfying
  \begin{equation*}
    \mathcal{L} \rbra{\rho} = \E_{s\sim \mu} \sbra{\e^{\ii H s}\rho \e^{-\ii H s}} - \rho.
  \end{equation*}
\end{theorem}

\begin{proof}
  Consider the eigendecomposition of $H$ as $H = \sum_j \lambda_j \ket{j}\bra{j}$.
  Let $\rho_{jk} = \bra{j} \rho \ket{k}$, and we have $\rho = \sum_{jk} \rho_{jk}\ket{j}\bra{k}$.
  Then, we can compute directly that
  \begin{equation*}
    \begin{aligned}
      \mathcal{L}\rbra{\rho}
      &= \int_{\mathbb{R}} \dd{s}\; p_{\mu}\rbra{s}  \sum_{jk}\e^{\mathrm{i}\rbra{\lambda_j-\lambda_k}s} \rho_{jk} \ket{j}\bra{k} - \sum_{jk} \rho_{jk} \ket{j}\bra{k} \\
      &= \sum_{jk} \rbra{\hat{\mu}\rbra{\lambda_j-\lambda_k} -1} \rho_{jk} \ket{j}\bra{k}.
    \end{aligned}
  \end{equation*}
  Therefore, we know
  \begin{equation}%
    \label{eqn:compound-poisson-dynamics}
    \e^{t \mathcal{L}} \rbra{\rho} =  \sum_{jk} \exp\rbra{t\rbra{\hat{\mu}\rbra{\lambda_j-\lambda_k} -1}} \rho_{jk} \ket{j}\bra{k},
  \end{equation}
  which is the same as $\Phi_{H, \mu_t} \rbra{\rho}$ by~\Cref{prop:hamiltonian-twirling-is-schur}.
\end{proof}

The implications of the above theorem are discussed as follows.

Let $X_j \sim \mu$ be independent and identically distributed random
variables, and $N\sim \Poisson\rbra{t}$.
Then, by~{\cite[Chapter XII.2]{Fel76}}, we know $S_N \coloneqq \sum_{j=1}^N X_j$
follows the distribution $\mu_t$.

First, when $\mu$ is the Dirac measure $\delta_{s_0}$, the Lindbladian is
\begin{equation*}
  \mathcal{L} \rbra{\rho} = \e^{\ii H s_0} \rho \e^{-\ii H s_0} - \rho.
\end{equation*}
This was also considered in~{\cite[Equation 6.26]{Kho91}}.
One can simulate this Lindbladian dynamics $\e^{t \mathcal{L}}$ by
implementing $\Phi_{\mu_t}$ in Hamiltonian simulation time $O(\abs{s_{0}t})$ in expectation
without ancilla.

Second, when $\mu = \sum_j p_j \delta_{s_j}$, the Lindbladian is
\begin{equation*}
  \mathcal{L} \rbra{\rho} = \sum_j p_j \e^{\ii H s_j} \rho \e^{-\ii H s_j} - \rho.
\end{equation*}
Suppose all $s_j\ge 0$, one can simulate this dynamics in Hamiltonian simulation
time $O(t\sum_j p_j s_j)$ in expectation.

In addition, \Cref{eqn:compound-poisson-dynamics} suggests that the non-diagonal
term of $\rho$ could decay not only exponentially with respect to $t$, but also
potentially more than exponentially with respect to the energy gap
$\lambda_j-\lambda_k$.
A typical case is when $\mu = \mathcal{N}\rbra{0,1}$ with characteristic
function $\hat{\mu}(z) = \exp\rbra{-\frac{z^2}{2}}$.
In this case, the decay is \emph{double exponential} with respect to the energy
gap.

\section{Lower Bounds on Query Complexity}%
\label{sec:lower-bound-evolution}

In this section, we prove lower bounds on the query complexity of simulating
the above Lindbladian dynamics.

For the query complexity, we need to specify the access model.
For this, we recall the definition of block-encoding, which is widely used
in the area of quantum algorithms.
\begin{definition}[Block-Encoding,~\cite{LC19,CGJ19,GSLW19}]\label{def:block-encoding}
  Let $A$ be an $s$-qubit linear operator. An $(s+a)$-qubit unitary $U$
  is called an $(\alpha, a, \varepsilon)$-block-encoding of $A$, if
  \begin{equation*}
    \norm*{A-\alpha \rbra*{\bra{0}^{\otimes a}\otimes I} U \rbra*{\ket{0}^{\otimes a}\otimes I}}\le \varepsilon.
 \end{equation*}
\end{definition}
If $\varepsilon = 0$ and $a$ is not necessary to mention in the context, we will
simply call it an $\alpha$-block-encoding.
If $\alpha = 1$, we will just call it a block-encoding.

The core technique to prove lower bounds is the sample-to-query lifting proposed in~\cite{WZ25} and recently improved in~\cite{CWZ25} via the samplizer
method~\cite{WZ25,WZ25b} and results in~\cite{TWZ25}.
It is a useful technique that can ``lift'' the sample complexity lower bound to
query complexity lower bound.
To be more exact, let $\rho_1$ and $\rho_2$ be two density operators on a same Hilbert
space.
Following~\cite{WZ25, CWZ25}, we will use $\SampCom\rbra{\mathsf{P}}$ to denote the
sample complexity of solving the quantum state testing problem of the property
$\mathsf{P}$ with independent and identical sample access to the states.
Specifically, $\SampCom\rbra{\Dis_{\rho_1,\rho_2}}$ is used to denote the sample
complexity to distinguish $\rho_1$ from $\rho_2$.
It is clear that $\SampCom\rbra{\Dis_{\rho_1,\rho_2}} = \Omega\rbra{1/\gamma}$ (see~{\cite[Lemma
  2.3]{WZ25}}), where $\gamma = 1 - \Fid\rbra{\rho_1, \rho_2}$ is the infidelity
between $\rho_1$ and $\rho_2$.
Similarly, we use $\QueryD\rbra{\mathsf{P}}$ to denote the query complexity of
solving the quantum testing problem of the property $\mathsf{P}$ with
$1$-block-encoded query access (originally $2$-block-encoding in~\cite{WZ25} and
improved in~\cite{CWZ25}).
Specifically, $\QueryD\rbra{\Dis_{\rho_1,\rho_2}}$ is used to denote the sample
complexity to distinguish $\rho_1$ from $\rho_2$, given a block-encoding access to $\rho$
where $\rho$ is either $\rho_1$ or $\rho_2$.

As shown in~\cite{WZ25,CWZ25}, $\QueryD\rbra{\Dis_{\rho_1,\rho_2}}$ can be characterized
by the infidelity between the density operators.
\begin{theorem}[Query Lower Bound from Lifting on Fidelity,~{\cite[Theorem 1.1]{CWZ25}} and~{\cite[Corollary
    1.3]{WZ25}}]%
  \label{thm:query-lower-bound-lifting-fidelity}
  Let $\rho$ and $\sigma$ be two quantum states on a same Hilbert space.
  Then
  $\QueryD\rbra{\Dis_{\rho_1,\rho_2}} =\Omega\rbra{1/\sqrt{\gamma}}$,
  where $\gamma = 1 - \Fid\rbra{\rho_1, \rho_2}$ is the infidelity between
  $\rho_1$ and $\rho_2$.
\end{theorem}

\subsection{Lower Bounds for Stable Distribution Twirling}

We first prove lower bounds for our stable distribution twirling algorithms. Since a Gaussian distribution is a stable distribution with $\alpha = 2$,
the theorem also include lower bounds for Gaussian twirling algorithms.
\begin{theorem}[Lindbladian Simulation Lower Bound for Stable Distribution Twirling]%
  \label{thm:lindbladian-simulation-lower-bound-stable}
  Let $1< \alpha\le 2$. For a Lindbladian $\mathcal{L}$ of the form
  \begin{equation}\label{eqn:Lindblad-stable-twirling}
    \mathcal{L}\rbra{\rho} = C_{\alpha} \int_{\abs{s}>0} \dd{s} \frac{1}{\abs{s}^{1+\alpha}} \rbra*{\mathrm{e}^{-\mathrm{i}Hs}\rho \mathrm{e}^{\mathrm{i}Hs} - \rho},
  \end{equation}
  where $H = \sum_j \lambda_j \ket{j}\bra{j}$ is a Hamiltonian in eigenbasis,  and $C_{\alpha}$ is a normalization constant such that
  \begin{equation*}
    \mathrm{e}^{t \mathcal{L}} \rbra*{\rho} = \sum_{jk} \mathrm{e}^{-t \abs{\lambda_j-\lambda_k}^{\alpha}} \rho_{jk}\ket{j}\bra{k},
  \end{equation*}
  Then, simulating the evolution of $\mathrm{e}^{t \mathcal{L}}$ requires $\Omega\rbra{t^{1/\alpha}}$ queries to the block-encoding of $H$.
\end{theorem}

\begin{proof}
  Let $c\ge 2$ be some number to be determined later.
  In the following, let $\sigma_1 = \frac{1}{2}I$ and $\sigma_2 = \rbra{\frac{1}{2}+\frac{1}{c}} \ket{0}\bra{0} + \rbra{\frac{1}{2}-\frac{1}{c}}\ket{1}\bra{1}$.
  Note that, the infidelity between $\sigma_1$ and $\sigma_2$ can be computed as
  \begin{equation*}
    \gamma \coloneqq 1- \Fid\rbra{\sigma_1, \sigma_2} = 1 - \frac{1}{2} \rbra*{\sqrt{1+\frac{2}{c}} + \sqrt{1-\frac{2}{c}} } \le \frac{4}{c^2}.
  \end{equation*}
  Therefore, by~\Cref{thm:query-lower-bound-lifting-fidelity}, we know that the query complexity lower bound for
  discriminating these states is $\Omega\rbra{c}$.

  We now show that, given query access to a block-encoding of $\sigma$ where $\sigma = \sigma_1$ or $\sigma_2$,
  performing Lindbladian simulation with the jump operator being $\sigma$
  for time $t = 9c^{2}$ for constant rounds suffices to distinguish the cases with probability at least $2/3$.
  The reduction is as follows.

\begin{algorithm}[H]
  \caption{Reduction in Stable Distribution Case: Distinguish between $\sigma_1$ and $\sigma_2$ via Lindbladian Simulation}%
  \label{alg:reduction-lifting-stable}
\begin{algorithmic}[1]
\Require Query access to a block-encoding of $H$, which is either $H_1=\sigma_1$ or $H_2=\sigma_2$.
\Ensure Decide the input being either $\sigma_1$ or $\sigma_2$.
\State{}$\rho_0 \gets \ket{+}\bra{+}$.
\State{}$t \gets 9c^{\alpha}$.
\State{}$\texttt{flag} \gets \texttt{false}$.
\State{}Let $\mathcal{L}$ denote the Lindbladian $\mathcal{L}$ in~\Cref{eqn:Lindblad-stable-twirling}.
\For{$i = 1$ to $100$}
    \State Prepare $\rho_0$.
    \State $\rho_t \gets e^{t\mathcal{L}}(\rho_0)$.
    \State Measure $\rho_t$ in $\{\ket{+},\ket{-}\}$ basis.
    \If{Measurement outcome is $\ket{-}$}
      \State $\texttt{flag} \gets \texttt{true}$.
    \EndIf
\EndFor

\If{$\texttt{flag} = \texttt{false}$}
    \State{} \textbf{Return} $\sigma_1$.
\Else{}
    \State{} \textbf{Return} $\sigma_2$.
\EndIf{}

\end{algorithmic}
\end{algorithm}

For the first case where $H = \sigma_1$, the channel is the identity channel,
and one can always get the outcome $\ket{+}$.
Therefore, in this case, the algorithm always succeeds.

For the second case where the Lindbladian has the single jump operator
$\sigma_{2}$, using~\Cref{prop:hamiltonian-twirling-is-schur}, we know
after the evolution, the state is just
\begin{equation*}
  \rho_{t} =
  \frac{1}{2}
  \begin{pmatrix}
    1 & \mathrm{e}^{-2t/c^{\alpha}} \\
    \mathrm{e}^{-2t/c^{\alpha}} & 1 \\
  \end{pmatrix}
  =
    \frac{1}{2}
  \begin{pmatrix}
    1 & \mathrm{e}^{-18} \\
    \mathrm{e}^{-18} & 1 \\
  \end{pmatrix}
\end{equation*}
Therefore, the probability of getting the measurement result $\ket{-}$ is just
$\tr(\rho_t\ket{-}\bra{-}) = \frac{1-e^{-18}}{2}$.
For repeating $100$ times, the result of getting at least one $\ket{-}$ is
$1-(\frac{1+e^{-18}}{2})^{100} \ge \frac{2}{3}$.
Therefore,~\Cref{alg:reduction-lifting-stable} can successfully distinguish these two
cases, meaning that one needs $\Omega\rbra{c} = \Omega\rbra{t^{1/\alpha}}$ queries to the
block-encoding of $H$.
Therefore, the theorem holds as we desired.
\end{proof}

\subsection{A Lower Bound for Poisson-Twirling Channels}
\begin{theorem}[Query Lower Bound for Lindbladian Evolution with Mixed Unitary Jump Operators]%
  \label{thm:lower-bound-lindbladian-Poisson}
  Let $\mathcal{L}$ denote the following Lindbladian operator 
  \begin{equation}
  \mathcal{L} \rbra{\rho} = \sum_{j=1}^N p_j \e^{\ii H s_j} \rho \e^{-\ii H s_j} - \rho,
  \end{equation}
  where $H$ is a Hamiltonian with $\norm{H}\le 1$, $p_j\ge 0$ and $\sum_j p_j =1$, and $s_j$ are fixed constants.
  Then, simulating the evolution of $\mathcal{L}$ for time $t$ requires $\Omega\rbra{t}$
  queries to the block-encoding of $H$, even for the case $N=1$ where 
  \begin{equation}\label{eqn:Lindblad-single-poisson}
  \mathcal{L} \rbra{\rho} = \e^{\ii H} \rho \e^{-\ii H} - \rho.
  \end{equation}
\end{theorem}

\begin{proof}
  Let $c\ge 2$ be some number to be determined later.
  In the following, let $\sigma_1 = \frac{1}{2}I$ and $\sigma_2 = \rbra{\frac{1}{2}+\frac{1}{c}} \ket{0}\bra{0} + \rbra{\frac{1}{2}-\frac{1}{c}}\ket{1}\bra{1}$.
  Note that, the infidelity between $\sigma_1$ and $\sigma_2$ can be computed as
  \begin{equation*}
    \gamma \coloneqq 1- \Fid\rbra{\rho_1, \rho_2} = 1 - \frac{1}{2} \rbra*{\sqrt{1+\frac{2}{c}} + \sqrt{1-\frac{2}{c}} } \le \frac{4}{c^2}.
  \end{equation*}
  Therefore, by~\Cref{thm:query-lower-bound-lifting-fidelity}, we know that the query complexity lower bound for
  discriminating these states is $\Omega\rbra{c}$.

  We now show that, given query access to a block-encoding of $\sigma$ where $\sigma = \sigma_1$ or $\sigma_2$,
  performing a Lindbladian simulation
  for time $t = \frac{\pi c}{4}$ for constant rounds suffices to distinguish the cases with probability at least $2/3$.
  The reduction is as follows.

\begin{algorithm}[H]
  \caption{Reduction in Poisson Case: Distinguish between $\sigma_1$ and $\sigma_2$ via Lindbladian Simulation}%
  \label{alg:reduction-lifting-Poisson}
\begin{algorithmic}[1]
\Require Query access to a block-encoding of $H$, which is either $H_1=\sigma_1$ or $H_2=\sigma_2$.
\Ensure Decide the input being either $\sigma_1$ or $\sigma_2$.
\State{}$\rho_0 \gets \ket{+}\bra{+}$.
\State{}$t \gets \frac{\pi c}{4}$.
\State{}$\texttt{flag} \gets \texttt{false}$.
\State{}Let $\mathcal{L}$ denote the Lindbladian with a single jump operator $\mathrm{e}^{-\mathrm{i}H}$.
\For{$i = 1$ to $100$}
    \State Prepare $\rho_0$.
    \State $\rho_t \gets e^{t\mathcal{L}}(\rho_0)$.
    \State Measure $\rho_t$ in $\{\ket{+},\ket{-}\}$ basis.
    \If{Measurement outcome is $\ket{-}$}
      \State $\texttt{flag} \gets \texttt{true}$.
    \EndIf
\EndFor

\If{$\texttt{flag} = \texttt{false}$}
    \State{} \textbf{Return} $\sigma_1$.
\Else{}
    \State{} \textbf{Return} $\sigma_2$.
\EndIf{}

\end{algorithmic}
\end{algorithm}

  For the first case where the Lindbladian has the single jump operator
$\sigma_1$, the channel is the identity channel, and one can always get
the outcome $\ket{+}$.
Therefore, when the input is $\rho_{1}$, the algorithm always succeeds.

For the second case where the Lindbladian has the single jump operator
$\sigma_{2}$, using~\Cref{prop:hamiltonian-twirling-is-schur}, we know
after the evolution, the state is just
\begin{equation*}
  \rho_{t} =
  \frac{1}{2}
  \begin{pmatrix}
    1 &  f\rbra{c} \\
   \bar{f}\rbra{c} & 1 \\
  \end{pmatrix}
\end{equation*}
where $f(c) = \mathrm{e}^{c\pi/4\rbra{\cos\rbra*{2/c}-1}} \mathrm{e}^{\mathrm{i} \pi c/4 \sin \rbra*{2/c}}$.
\end{proof}
 Then, after measuring in $\ket{+}, \ket{-}$ basis, the result of getting $\ket{-}$ is just 
 \begin{equation*}
   p_{-}\rbra{c} =  \frac{1}{2} \rbra*{1-\Re{f(c)}} = \frac{1}{2}\rbra{1-\mathrm{e}^{\pi c/4\rbra*{\cos\rbra*{2/c}-1}}  \cos \rbra*{\pi c/4 \sin \rbra*{2/c}}}.
 \end{equation*}
 Note that 
 \begin{equation*}
    \lim_{c\to +\infty} p_{-}\rbra{c} = \frac{1-\cos\rbra{\pi/2}}{2} = \frac{1}{2}.
 \end{equation*}
 Thus, for sufficiently large $t$, we can assume $p_{-}\rbra{t} \ge 0.49$.
 For repeating $100$ times, the result of getting at least one $\ket{-}$ is no less than $1-0.51^{100}\ge 2/3$.
 Therefore,~\Cref{alg:reduction-lifting-Poisson} can successfully distinguish these two
cases, meaning that one needs $\Omega\rbra{c} = \Omega\rbra{t}$ queries to the
block-encoding of $H$.
Therefore, the theorem holds as we desired.

\end{document}